\documentclass[conference]{IEEEtran}

\usepackage{amsmath, amsthm, amssymb}
\usepackage{latexsym}
\usepackage{graphicx}
\usepackage{verbatim}
\usepackage{mathrsfs}
\usepackage{algorithmic}
\usepackage{float}
\usepackage[lofdepth, lotdepth]{subfig}
\usepackage{array}
\usepackage{url}
\usepackage{blkarray, multirow}
\usepackage{color}

\theoremstyle{plain}
\newtheorem{theorem}{Theorem}
\newtheorem{lemma}[theorem]{Lemma}

\theoremstyle{definition}
\newtheorem{definition}[theorem]{Definition}
\newtheorem{example}[theorem]{Example}
\newtheorem{remark}[theorem]{Remark}

\newcommand{\D}{{\mathcal D}}

\newcommand{\bM}{{\boldsymbol{M}}} 
 
\newcommand{\bI}{{\boldsymbol I}} 
 
\newcommand{\bP}{{\boldsymbol P}} 
\newcommand{\bQ}{{\boldsymbol Q}} 

\newcommand{\bF}{{\boldsymbol F}} 
\newcommand{\bH}{{\boldsymbol H}} 
\newcommand{\bK}{{\boldsymbol K}}

\newcommand{\bA}{{\boldsymbol A}}
\newcommand{\bB}{{\boldsymbol B}}
\newcommand{\bC}{{\boldsymbol C}}
\newcommand{\bD}{{\boldsymbol D}}
\newcommand{\bE}{{\boldsymbol E}}

\newcommand{\bS}{\boldsymbol{S}}
\newcommand{\bT}{\boldsymbol{T}}

\newcommand{\ba}{{\boldsymbol a}}

\newcommand{\bu}{{\boldsymbol u}}
\newcommand{\bv}{{\boldsymbol v}}

\newcommand{\bs}{{\boldsymbol s}}

\newcommand{\bc}{{\boldsymbol c}}
\newcommand{\bO}{{\boldsymbol O}}
\newcommand{\bzero}{{\boldsymbol 0}}

\newcommand{\br}{{\boldsymbol{r}}}
\newcommand{\bx}{{\boldsymbol{x}}}

\newcommand{\bG}{{\boldsymbol{G}}}
\newcommand{\be}{{\boldsymbol e}}

\newcommand{\XiG}{\boldsymbol{\Xi_{\boldsymbol{G}''}}}
\newcommand{\XiA}{\boldsymbol{\Xi_{\boldsymbol{A}''}}}
\newcommand{\XiB}{\boldsymbol{\Xi_{\boldsymbol{B}''}}}
\newcommand{\XiC}{\boldsymbol{\Xi_{\boldsymbol{C}''}}}
\newcommand{\XiD}{\boldsymbol{\Xi_{\boldsymbol{D}}}}
\newcommand{\XiE}{\boldsymbol{\Xi_{\boldsymbol{E}}}}
\newcommand{\XiF}{\boldsymbol{\Xi_{\boldsymbol{F}}}}
\newcommand{\XiHi}{\boldsymbol{\Xi_{\boldsymbol{H}_i}}}

\newcommand{\XiBo}{\boldsymbol{\Xi_{\boldsymbol{B}''_1}}}
\newcommand{\XiBt}{\boldsymbol{\Xi_{\boldsymbol{B}''_2}}}
\newcommand{\XiBi}{\boldsymbol{\Xi_{\boldsymbol{B}''_i}}}
\newcommand{\XiBr}{\boldsymbol{\Xi_{\boldsymbol{B}''_{(k-\mu)/p}}}}

\newcommand{\XiEo}{\boldsymbol{\Xi_{\boldsymbol{E}_1}}}
\newcommand{\XiEt}{\boldsymbol{\Xi_{\boldsymbol{E}_2}}}
\newcommand{\XiEi}{\boldsymbol{\Xi_{\boldsymbol{E}_i}}}
\newcommand{\XiEr}{\boldsymbol{\Xi_{\boldsymbol{E}_{(k-\mu)/p}}}}

\newcommand{\ff}{\mathbb{F}}

\newcommand{\fq}{\mathbb{F}_q}

\newcommand{\fs}{f_\text{super}}
\newcommand{\fmp}{f_{(\mu,p)}}
\newcommand{\bst}{\tilde{\boldsymbol{s}}}

\newcommand{\vt}{\tilde{v}}
\newcommand{\brt}{\tilde{\boldsymbol{r}}}

\newcommand{\al}{\alpha}
\newcommand{\bt}{\beta}

\newcommand{\entropy}{{\sf H}}

\newcommand{\sgn}{{\sf sgn}}
\newcommand{\dist}{{\mathsf{d}}}
\newcommand{\weight}{{\mathsf{wt}}}

\newcommand{\cC}{{\mathscr{C}}}

\newcommand{\nin}{\noindent}

\newcommand{\et}{{\emph{et al.}}}

\begin{document}
\pagestyle{empty}

\title{Secure Erasure Codes With Partial Decodability}
 \author{
   \IEEEauthorblockN{
     Son Hoang Dau\IEEEauthorrefmark{1},
     Wentu Song\IEEEauthorrefmark{2}, 
     Chau Yuen\IEEEauthorrefmark{3}
		} 
   \IEEEauthorblockA{
   Singapore University of Technology and Design, Singapore\\     
		Emails: $\{${\it\IEEEauthorrefmark{1}sonhoang\_dau, 
		\IEEEauthorrefmark{2}wentu\_song,
		\IEEEauthorrefmark{3}yuenchau}$\}$@sutd.edu.sg		
		}
 }
\maketitle

\begin{abstract}
The MDS property (aka the $k$-out-of-$n$ property) 
requires that if a file is split into several symbols and subsequently encoded into 
$n$ coded symbols, each being stored in one storage node of a distributed storage system (DSS), 
then an user can recover the file by accessing any $k$ nodes.  
We study the so-called \emph{$p$-decodable $\mu$-secure} erasure coding scheme $(1 \leq p \leq k - \mu, 
0 \leq \mu < k, p | (k-\mu))$, which satisfies the MDS property and the following additional properties: 
\begin{itemize}
	\item[(P1)] strongly secure up to a threshold: an adversary which eavesdrops at most $\mu$ storage nodes
	gains no information (in Shannon's sense) about the stored file,
	\item[(P2)] partially decodable: a legitimate user can recover a subset of $p$ file symbols by accessing some 
	$\mu + p$ storage nodes. 
\end{itemize}
The scheme is \emph{perfectly} $p$-decodable $\mu$-secure if it satisfies 
the following additional property:
\begin{itemize}
	\item[(P3)] weakly secure up to a threshold: an adversary which eavesdrops more than $\mu$ but 
	less than $\mu+p$ storage nodes cannot reconstruct any part of the file.
\end{itemize}
Most of the related work in the literature only focused on the case $p = k - \mu$. 
In other words, no partial decodability is provided: 
an user cannot retrieve any part of the file by accessing less than $k$ nodes. 
For our more general code, Property (P2) guarantees partial decodability: once the user accesses 
$p$ more nodes than the strong security threshold $\mu$, 
it can start to decode some subset of $p$ file symbols. 

We provide an explicit construction of $p$-decodable $\mu$-secure coding schemes over small fields 
for all $\mu$ and $p$. 
That construction also produces \emph{perfectly} $p$-decodable $\mu$-secure schemes over small fields
when $p = 1$ (for every $\mu$), and when $\mu = 0, 1$ (for every $p$).
We establish that perfect schemes exist over \emph{sufficiently large} fields for almost all $\mu$ and $p$.   
\end{abstract}

\section{Introduction}
\label{sec:intro}

Data replication is the most common way for distributed storage systems (DSS) to guarantee
high data availability and node failure tolerance. Most of the current distributed storage systems
are using $3$-way replication where each piece of data is replicated three times and each of its copy
is stored at a different storage node in the system. If at most two storage nodes are down, the data
is still available at at least one node.
However, $3$-way replication is highly inefficient in storage overhead, as only a very modest portion $33\%$ 
of the whole storage capacity can be used.
As the demand for data storage scales up quickly, replication based storage
systems incur significantly high cost in terms of storage footprint and power usage for cooling systems.
It is well known that erasure codes~\cite{MW_S} possess lots of advantages over replication~\cite{WeatherspoonKubiatowicz2002}. 
Giants such as Microsoft, Facebook, and Google have, therefore, included erasure codes, alongside
replication, in their 
distributed storage systems~\cite{Huang2012, Thusoo2010, 
Ford2010}.    

%%%%%%%%%%%%%%%%%%%%%%%%%%%%%%
\begin{figure}[t]
\centering
\includegraphics[scale=0.9]{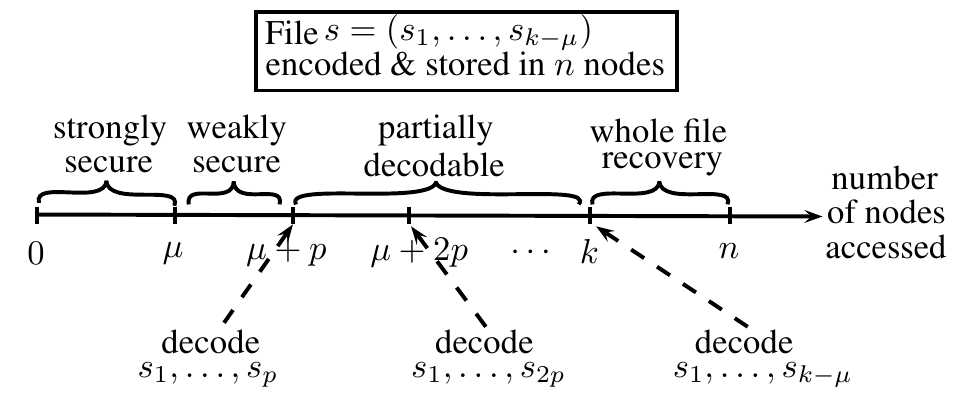}
\caption{Illustration for perfectly $p$-decodable $\mu$-secure erasure coding scheme.
Partial decoding starts when the first $\mu+p$ nodes are accessed. At this point, the first
subset of $p$ file symbols $(s_1, \ldots, s_p)$ is reconstructed. When the first $\mu + 2p$
nodes are accessed, the next subset of $p$ file symbols $(s_{p+1}, \ldots, s_{2p})$ can also be
decoded. In fact, if the user requests only $(s_{p+1}, \ldots, s_{2p})$, it is sufficient to
access the first $\mu$ nodes and the nodes numbered from $\mu+p+1$ to $\mu + 2p$.
The general requirement for partial decoding is: $(s_{rp+1},\ldots, s_{(r+1)p})$ can be
reconstructed by accessing the first $\mu$ nodes and $p$ additional nodes numbered from $\mu+rp+1$
to $\mu + (r+1)p$. We always assume that $p|(k-\mu)$.}
\label{fig:illustration}
\vspace{-15pt}
\end{figure}
%%%%%%%%%%%%%%%%%%%%%%%%%%%%%%%%%%%
In this work we investigate a construction of erasure coding schemes for DSS that are secure
in terms of data confidentiality: even when some of its storage nodes are compromised or 
eavesdropped by some unwanted party, the stored file is still kept confidential.     
Moreover, such secure codes must provide easy partial decoding for a legitimate
user, which can often access more nodes than an illegal adversary.  
We proposed the so-called \emph{$p$-decodable $\mu$-secure} erasure coding scheme $(1 \leq p \leq k - \mu, 
0 \leq \mu < k)$, which satisfies the MDS property (the file can be reconstructed by accessing any $k$ out of $n$ 
storage nodes) and the following additional properties: 
\begin{itemize}
	\item[(P1)] \emph{strongly secure} up to a threshold: an adversary which eavesdrops at most $\mu$ storage nodes
	gain no information (in Shannon's sense) about the stored file,
	\item[(P2)] \emph{partially decodable}: a legitimate user can recover a subset of $p$ file symbols by accessing some 
	$\mu + p$ storage nodes. 
\end{itemize}
Regarding (P2), throughout this paper we always assume that $p | (k - \mu)$. 
In other words, we can always partition the set of $k-\mu$ file symbols into 
subsets of size $p$ each. 
Apart from (P1)-(P2), if the following additional property is also satisfied, the scheme is referred to as \emph{perfectly} $p$-decodable $\mu$-secure:
\begin{itemize}
	\item[(P3)] \emph{weakly secure} up to a threshold: an adversary which eavesdrops more than $\mu$ but 
	less than $\mu+p$ storage nodes cannot reconstruct any part of the file,
\end{itemize}
We illustrate the properties of a perfect coding scheme in Fig~\ref{fig:illustration}. 
%%%%%%%%%%%%%%%%%%%%%%%%%%%%%%%%
\begin{table*} 
\begin{tabular}{|c|c|c|c|c|}
\hline
& Code rate & Strong security threshold & Weak security threshold &  Partial decodability threshold\\
\hline
Systematic erasure code & $k/n$ & $0$ & 0 & $1$\\
\hline
Codes for (erasure) wiretap channel II~\cite{OzarowWyner1984,SubramanianMcLaughlin2009}
& $(k-\mu)/n$ & $\mu$ & N.A. & N.A.\\
\hline
Ramp secret sharing scheme~\cite{Yamamoto1986, Yamamoto2006, 
OliveiraLimaVinhozaBarrosMedard2012}
& $(k-\mu)/n$ & $\mu$ & $k-1$ & $k$\\
\hline
\textbf{$p$-decodable $\mu$-secure
code} & $(k-\mu)/n$ & $\mu$ & N.A.
& $\mu + p$ \\
\hline
\textbf{Perfectly $p$-decodable $\mu$-secure
code} & $(k-\mu)/n$ & $\mu$ & $\mu + p - 1$
& $\mu + p$ \\
\hline
\end{tabular}
\caption{Comparison among five erasure coding schemes. Strong (weak) security
threshold refers to the maximum number of storage nodes that the adversary is allowed to
access without jeopardizing the strong (weak) security of the scheme.
Partial decodability threshold refers to the number of nodes an user
has to access to start decoding the file partially. 
Here $0 \leq \mu \leq k-1$, $1 \leq p \leq k - \mu$, and $p | (k - \mu)$.
When $p = 1$ and $\mu=0$, the (perfectly) $p$-decodable $\mu$-secure scheme is actually systematic. An `N.A.' entry means that the corresponding
threshold can take any value and that threshold is not even considered
in the design of the coding scheme.}
\vspace{-10pt}
\end{table*}

Note that when $p = k - \mu$, only (strong and weak) security is guaranteed and 
there is no partial decodability: 
an user cannot retrieve any part of the file by accessing less than $k$ nodes. 
Such a secure coding scheme was first studied in the work of 
Yamamoto~\cite{Yamamoto1986} in the context of ramp secret sharing scheme. 
Recently, superregular matrices (all square submatrices are invertible) such as Cauchy matrices 
have also been employed to construct such codes~\cite{Yamamoto2006, OliveiraLimaVinhozaBarrosMedard2012}.  
Similar work in secure regenerating codes, which can be regarded as vector erasure codes with optimal 
node repair, can be found in~\cite{DauSongYuen2014, KadheSprintson2014}.
%The concept of weak security, in constrast to strong security, was also discovered by 
%Bhattad and Narayanan~\cite{BhattadNarayanan2005} in a more general context of network coding. 

In this work we address the gap in the literature regarding partial decodability of 
erasure coding schemes.  
In our proposed coding scheme, while Property (P1) and Property (P3) provide (strong and weak)
security, Property (P2) guarantees partial decodability: once the user accesses $p$ more nodes 
than the strong security threshold $\mu$, it can start to decode some subset of $p$ file symbols. 
A secure erasure code that supports partial decoding is particularly useful in applications
involving retrieval of large files. A typical example is in video streaming services, where
a large-size video is often split into chunks and these chunks are then streamed one by one
to the user. Our proposed coding scheme, if employed in such services, would provide
not only confidentiality but also ease of partial retrieving of the video to any desired level.  
We stress that extending the existing secure coding schemes 
by taking into account partial decodability does not result in any overhead.  
%%%%%%%%%%%%%%%%%%%%%%%%%%%%%%%%%
\begin{table}[H]
\centering
\begin{tabular}{|l|l|l|}
\hline
& Small field & Large field\\
\hline
\multirow{1}{*}{$p$-decodable $\mu$-secure} & $^*\forall \mu, \forall p$ & N.A. \\
\hline
\multirow{3}{*}{perfectly $p$-decodable $\mu$-secure} & $p = k - \mu$ \cite{Yamamoto1986, Yamamoto2006, OliveiraLimaVinhozaBarrosMedard2012} & 
\multirow{3}{*}{$^*\forall \mu, \forall p$}\\
 & $^*p = 1$, $\forall \mu$  & \\
& $^*\mu = 0, 1$, $\forall p$ & \\
\hline
\end{tabular}
\caption{A summary of existence of (perfectly) $p$-decodable $\mu$-secure
coding schemes. The entries with an asterisk `$^*$' are the new results established in
this work. A field is ``small" if its size is a polynomial in $n$ and $k$, or ``large'' if otherwise.}
\vspace{-10pt}
\end{table}
%%%%%%%%%%%%%%%%%%%%%%%%%%%%%%%%%%%%%%%%%%%%%
Our main contribution is summarized below. 
\begin{itemize}
	\item We provide an explicit construction over small fields for $p$-decodable 
	$\mu$-secure coding schemes for any $p$ and $\mu$. 
	\item We provide an explicit construction over small fields for \emph{perfectly} $p$-decodable 
	$\mu$-secure coding schemes when $p = 1$ (for every $\mu$), and when $\mu = 0,1$ (for every $p$, 
	$p | k - \mu$).
	\item We prove that perfectly $p$-decodable $\mu$-secure schemes exist over sufficiently
	large fields for almost all $p$ and $\mu$. 
\end{itemize}
The paper is organized as follows. Section~\ref{sec:pre} provides necessary notation and definitions. 
The construction of $p$-decodable $\mu$-secure schemes is presented in Section~\ref{sec:imperfect}. 
We discuss the existence of perfect coding schemes in Section~\ref{sec:perfect}. 
The paper is concluded in Section~\ref{sec:conclusion}. 

\section{Preliminaries}
\label{sec:pre}

\subsection{Notation}
\label{subsec:notation}

Let $\fq$ denote the finite field with $q$ elements. 
Let $[n]$ denote the set $\{1,2,\ldots,n\}$. 
%The \emph{support} of a vector $\bu = (u_1, \ldots, u_n) \in \fq^n$ is defined by
%$\supp(\bu) = \{i \in [n]:\ u_i \neq 0\}$.
For a vector $\bu = (u_1,\ldots,u_n) \in \fq^n$, let $\bu^{\text{T}}$ denote the transpose of $\bu$. 
For any $n \geq 1$ let $\bI_n$ denote the identity 
matrix of order $n$.
We also use $\bO$ to denote the all-zero matrix
of certain size. 

Below we define standard notation from coding theory (see \cite{MW_S}).
The (Hamming) \emph{weight} of $\bu$ is $\weight(\bu) = |\{i: u_i \neq 0\}|$. 
The (Hamming) \emph{distance} between two vectors $\bu$ and $\bv$ is $\dist(\bu,\bv) = \weight(\bu-\bv)$. 
An $[n,k,d]_q$ (error-correcting) code $\cC$ (sometimes $d$ and $q$ are dropped) is a subspace of the
vector space $\fq^n$ of dimension $k$ such that the minimum distance between any two distinct vectors
in that subspace is at least $d$. We called $d =\dist(\cC)$ the \emph{minimum distance} of the code. 
The well-known Singleton bound states that for any $[n, k, d]_q$ code, $d \leq n - k + 1$. A
code attaining the Singleton bound is called maximum distance separable (MDS).  
Any vector in a code is referred to as a codeword. 
A generator matrix of an $[n,k]_q$ code
is a $k \times n$ matrix whose rows are linearly independent codewords of the code. 

\subsection{Coset Coding Technique}
\label{subsec:coset_code}

%Let $\bR = (R_1, R_2, \ldots, R_\mu)$, where $R_i$'s $(i \in [\mu])$ are 
%independent and identically uniformly distributed random variables over $\fq$.  
%Let $\br = (r_1, r_2, \ldots, r_\mu)$ be a realization of $\bR$. 
Let $\br = (r_1, \ldots, r_\mu)$ be a vector of independent and 
identically uniformly distributed random variables over $\fq$.  
Let $\bS = (S_1, S_2, \ldots, S_{k-\mu})$, where $S_i$'s $(i \in [k-\mu])$ are 
independent and identically uniformly distributed random variables over some alphabet $\fq$.  
We assume that the file to be stored in the system is $\bs = (s_1, s_2, \ldots, s_{k-\mu}) \in \fq^{k-\mu}$, 
a realization of $\bS$.  
We call $k-\mu$ the file size and each $s_i$ a file symbol.
%In this work we always set $t = k - \mu$. 
  
We denote by $\D(n,k)$ a DSS with $n$ storage nodes where the file
can be recovered from the contents of any $k$ out of $n$ nodes.
Node $i$ $(i \in [n])$ stores a coded symbol $c_i$, which is a function of the file symbols $s_i$'s
and the random symbols $r_j$'s. 
Let $\bc = (c_1, c_2, \ldots, c_n)$. 
Let $\bC = (C_1, C_2, \ldots, C_n)$ be $\bc$'s corresponding vector of random variables over $\fq$.   
We only consider here scalar linear erasure coding schemes, based on $[n, k]_q$ MDS codes,
described as follows. 
The file $\bs \in \fq^{k-\mu}$ is encoded to $\bc = \bx \bG \in \fq^n$, where
\begin{itemize}
	\item $\bx = (\bs \mid \br)$ is obtained by concatenating $\bs$ and $\br$,
	\item $\bG$ is a generator matrix of an $[n,k, n-k+1]_q$ MDS code. We
	often refer to $\bG$ as the \emph{coding matrix}. 
\end{itemize}
It is well known in coding theory~\cite{MW_S} that if Node $i$ stores $c_i$ produced as above then the file 
can be recovered by accessing any $k$ nodes. 
For the coding scheme to be strongly secure against an adversary that can
access $\mu$ nodes (see Definition~\ref{def:strong_security}), the last $\mu$ rows of $\bG$ must generate an $[n, \mu]$
MDS code.  
In fact, this is an equivalent way to describe the coset coding technique invented by 
Ozarow and Wyner~\cite{OzarowWyner1984}. This technique has been widely adopted
in the network coding literature to secure a network code against a wiretapper (see, 
for instance~\cite{CaiChan2011} and references therein). 

%Extension of this work to regenerating codes~\cite{DimakisGodfreyWainwrightRamchandran2007}
%is discussed in Section Conclusion. 

\subsection{Security and Partial Decodability}
\label{subsec:security_decodability}

We assume that an adversary can eavesdrop/access any $m$ storage nodes of its choice
and tries to learn illegally the content of the stored file. 
We refer to $m$ as the adversary's \emph{strength}. 
In the following definitions, recall that $\bS = (S_1, S_2, \ldots, S_{k-\mu})$ represents the file stored in the system.  

\begin{definition}
\label{def:strong_security}
An erasure coding scheme for a DSS $\D(n,k)$ is called \emph{strongly secure}
against an adversary of strength $m$ $(m < k)$ if the entropy
\[
\entropy\big(\bS \mid \{C_i: i \in E\}\big) = \entropy \big( \bS \big),
\] 
for all subsets $E \subseteq [n]$, $|E| \leq m$. We also refer to such a coding scheme as \emph{strongly} $m$-secure. 
\end{definition}

In words, a coding scheme is strongly $\mu$-secure if an adversary 
which can access an arbitrary set of at most $\mu$ storage nodes
cannot obtain any information at all about the stored file. 
It is well-known that as long as the bottom $\mu \times n$ submatrix of
$\bG$ also generates an $[n, \mu]$ MDS code then the coding scheme described in 
Section~\ref{subsec:coset_code} is strongly $\mu$-secure. 
The MDS code generated by such a matrix $\bG$ is often called a \emph{nested} MDS code in the literature. 

\begin{definition}
\label{def:weak_security}
An erasure coding scheme for a DSS $\D(n,k)$ is called \emph{weakly secure}
against an adversary of strength $m$ $(m < k)$ if the entropy
\[
\entropy\big(S_j \mid \{C_i: i \in E \}\big) = \entropy\big(S_j\big),
\] 
for all $j \in [k-\mu]$ and for all subsets $E \subseteq [n]$, $|E| \leq m$.
We also refer to such a coding scheme as \emph{weakly} $m$-secure.
\end{definition}
 
The following lemma specifies a necessary and sufficient condition for the weak security of
the erasure coding scheme described in Section~\ref{subsec:coset_code}.  

\begin{lemma}
\label{lem:basic}
Let $\bs= (s_1,s_2,\ldots,s_{k-\mu}) \in \fq^{k-\mu}$ be the stored file and $\br = (r_1, r_2, \ldots, r_\mu)$
be some random vector over $\fq$. Let $(\bs \mid \br) \bM$, where $\bM$ is a $k \times m$ matrix over $\fq$, 
represent the $m$ coded symbols stored at some $m$ storage nodes that the adversary has access to. 
Then the adversary cannot determine any particular file symbol $s_i$ if and only if the column space of $\bM$
does not contain $\be_i$ for every $i \in [k-\mu]$, where $\be_i$ is the unit vector
with only one nonzero coordinate at the $i$th position. 
\end{lemma}
\begin{proof}
Appendix~\ref{app:1}. 
\end{proof}
\vskip 10pt 

The intuition behind the proof of this lemma is explained below. 
As the adversary obtains $(\bs \mid \br) \bM$, it can linearly transform these coded symbols
by considering the product $(\bs \mid \br) \bM \boldsymbol{\alpha}^{\text{T}}$, 
where $\boldsymbol{\alpha} \in \fq^m$ is some coefficient vector. 
The adversary can determine a file symbol $s_i$ $(i \in [k-\mu])$ if and only if there exists $\boldsymbol{\alpha}$
so that $\bM \boldsymbol{\alpha}^{\text{T}} = \be_i$. 
As $\bM \boldsymbol{\alpha}^{\text{T}}$ is a vector of $\cC^{\bM}$, we derive the conclusion
in Lemma~\ref{lem:basic}. 

After Yamamoto~\cite{Yamamoto1986}, the concept of weak security was also discovered by Bhattad and Narayanan~\cite{BhattadNarayanan2005} in a more general context of network coding. 
Weak security is important in practice since it guarantees that no meaningful 
information is leaked to the adversary, and often requires no additional overhead. 
For instance, suppose that the file $\bs = (s_1, s_2, s_3)$ is to be stored in a DSS
with four storage nodes, which tolerates one node failure. 
Using an usual systematic erasure code, the four nodes store $\bc = (s_1, s_2, s_3, s_1 + s_2 + s_3)$. 
However, if an adversary can access any node among the first three, then 
it can retrieve some $s_i$, which is part of the file. 
On the other hand, if the file is encoded into $\bc = (s_1+s_2, 
s_1 + s_3, s_2 + s_3, s_1 + s_2 + s_3)$, then an adversary who accesses one storage node would not
be able to determine any $s_i$. Indeed, for instance if it observes $s_1+s_2$, then 
it cannot determine the exact value of either $s_1$ or $s_2$, as for the adversary, both 
$s_1$ and $s_2$ are completely random variables.  
If $\bs$ is a video and $s_i$'s are movie chunks, then by using the latter coding scheme, 
an adversary who observes one storage node cannot determine each chunk, and hence, 
cannot play any part of the movie. Such coding scheme is said to be weakly secure against
an adversary of strength one, or weakly $1$-secure. 
Moreover, that coding scheme consists of the same number of storage nodes and can also
tolerate one node failure, hence introduces no storage overhead compared to a normal systematic
code. In fact, while strong security always comes with a cost in storage capacity, 
weak security is often given for free.  

\begin{definition} 
\label{def:security_decodability}
We consider the coding scheme described in Section~\ref{subsec:coset_code}, where
$\bG$ is a generator matrix of an $[n,k, n-k+1]_q$ MDS code and the file $\bs \in \fq^{k-\mu}$ is encoded into 
the coded vector $\bc = (\bs \mid \br) \bG \in \fq^n$. 
Suppose that $0 \leq \mu < k$, $1 \leq p \leq k - \mu$, and moreover $p | (k-\mu)$.  
The coding scheme based on $\bG$ is \emph{$p$-decodable $\mu$-secure} 
if it satisfies the following properties simultaneously. 
\begin{itemize}
	\item[(P1)] It is strongly $\mu$-secure as defined in Definition~\ref{def:strong_security}. 
	\item[(P2)] It is $p$-decodable: each subset of $p$ file symbols $\{s_{rp+1}, s_{rp+2}, \ldots, s_{(r+1)p}\}$ $(0 \leq r \leq (k-\mu)/p - 1)$
	can be reconstructed from the content of some $\mu + p$ storage nodes.  
\end{itemize}
The coding scheme is \emph{perfectly $p$-decodable $\mu$-secure} if it 
satisfies the following additional property:
\begin{itemize}
  \item[(P3)] It is weakly $(\mu+p-1)$-secure as defined in Definition~\ref{def:weak_security}.
\end{itemize}
We also say that the corresponding coding matrix $\bG$ is (perfectly) $p$-decodable $\mu$-secure. 
\end{definition} 
\vskip 10pt 

\begin{remark}
\mbox{}
\begin{itemize}
	\item A (perfectly) $1$-decodable $0$-secure code is simply a systematic code in the classical sense, where 
	each file symbol is stored as is (in its clear form) at some node.
	\item A perfectly $(k-\mu)$-decodable $\mu$-secure code is the type of secure codes studied in the work of 
	Yamamoto~\cite{Yamamoto2006} and Olivera {\et}~\cite{OliveiraLimaVinhozaBarrosMedard2012}.
	\item We can replace the weak security in (P3) by the \emph{block security} (equivalently, security against guessing)~\cite{BhattadNarayanan2005, DauSongYuen2014} which requires a stronger
	condition that
no information about any \emph{subset} of file symbols up to a certain 
size is known to the adversary. In fact, all of our results
in this work can be extended to block security. However, to simplify
the presentation, we restrict ourselves to only weak security. 
\end{itemize}
\end{remark}

We illustrate the concept of perfectly $p$-decodable $\mu$-secure coding scheme in the following example.

\begin{example}
\label{ex:illustration}
Let $k = 5$, $n = 6$, $\mu = 1$, $p = 2$, and $q = 11$. Let $\bs = (s_1, s_2, s_3, s_4) \in \ff^4_{11}$ be the stored file and $r$ a random symbol over $\ff_{11}$.
We use the following coding matrix
\[
\bG = 
\begin{pmatrix}
0 & 1 & 6 & 0 & 0 & 7\\
0 & 6 & 4 & 0 & 0 & 0\\
0 & 0 & 0 & 2 & 10 & 7\\
0 & 0 & 0 & 7 & 1 & 1\\
4 & 5 & 10 & 1 & 9 & 6.
\end{pmatrix}
\]   
The coding scheme is 
\[
\bc = (\bs \mid r) \bG. 
\]
We show below that this coding scheme is perfectly $2$-decodable $1$-secure. 
Firstly, it is easy to verify that $\bG$ generates a $[6,5]_{11}$ MDS code. Therefore, 
the file can be reconstructed by accessing any five storage nodes. 
The bottom $1 \times 6$ matrix also obviously generates a $[6,1]$ MDS code,
hence guarantees that the scheme is strongly $1$-secure. 
Hence (P1) is satisfied. 
To verify (P3), note that $\mu + p - 1 = 1 + 2 -1 = 2$. 
We can easily verify that any two columns of $\bG$ do not generate an 
unit vector $\be_i$ for every $i \in [4]$. Hence, according to 
Lemma~\ref{lem:basic}, the coding scheme is weakly $2$-secure.  
Lastly, for (P2), we prove that each of the 2-subsets $\{s_1,s_2\}$ and
$\{s_3,s_4\}$ can be reconstructed by accessing some three nodes $(\mu + p = 3)$.
Indeed, by accessing the first three nodes, an user obtains the product
\[
(s_1, s_2, s_3, s_4, r) 
\begin{pmatrix}
0 & 1 & 6\\
0 & 6 & 4\\
0 & 0 & 0\\
0 & 0 & 0\\
4 & 5 & 10. 
\end{pmatrix}
\]  
If we post-multiply the above product with the vector $\al = (10, 4, 5)^{\text{T}}$
then we obtain $s_1$. If we post-multiply the above product with the vector
$\beta = (5, 5, 1)^{\text{T}}$ then we obtain $s_2$. Hence, $\{s_1,s_2\}$
can be reconstructed by accessing the first three nodes. Similarly, we can verify
that $\{s_3, s_4\}$ can be reconstructed by accessing Node 1, Node 4, and Node 5. 
\end{example}  

\section{A Construction of $p$-Decodable $\mu$-Secure Coding Schemes}
\label{sec:imperfect}

We establish in this section a general construction 
for $p$-decodable $\mu$-secure erasure coding schemes, which satisfy
(P1) and (P2): an adversary which can access the content of $\mu$ storage nodes 
gains no information about stored file, and a legitimate user can retrieve each subset
of $p$ file symbols by accessing some $\mu + p$ storage nodes.  

\subsection{A General Construction}
\label{subsec:general}

We start with the definition of superregular matrices, which is critical in our construction. 

\begin{definition}[\cite{RothLempel1989}] 
A \emph{superregular} matrix is a matrix where every square submatrix is invertible.  
\end{definition} 

Two well-known constructions of superregular matrices are via Cauchy matrices and 
Vandermonde matrices~\cite{LacanFimes2004}. A Cauchy matrix is a matrix of the form $\bC = \big(1/(x_i - y_j)\big)_{i,j}$
where $x_i$'s and $y_j$'s are distinct elements of any finite field $\fq$.
Cauchy matrices are superregular by themselves. 
%A matrix $\bV = (\al_j^{i-1})_{i,j}$ where $\al_i$'s are all nonzero and distinct is called a
%Vandermonde matrix. Vandermonde matrices are not always superregular. However, if 
%$\bA = (\al_j^{i-1})_{i,j}$ of size $n \times k$ and $\bB = (\beta_j^{i-1})_{i,j}$
%of size $n \times n$ $n \geq k$, where $\al_i$'s and $\beta_i$'s are nozero and all distict,
%then $\bA\bB^{-1}$ is a superregular matrix~\cite{}. 
The following straightforward results about superregular matrices are especially useful
in our construction.    

\begin{lemma}
\label{lem:superregular}
Let $\bC$ be a superregular $k \times n$ matrix $(k \leq n)$.
Then the following hold.  
\begin{itemize}
  \item Any submatrix of $\bC$ is also superregular. %is again a superregular matrix. 
	\item Any subset of $k'$ rows of $\bC$ $(1 \leq k' \leq k)$
	generates an $[n, k']$ MDS code. Hence, 
	every nontrivial vector generated by these $k'$ rows has weight
	at least $n - k' + 1$. 
	\item Any subset of $n'$ columns of $\bC$ $(1 \leq n' \leq k)$, 
	generates a $[k, n']$ MDS code. Hence, every 
	nontrivial vector generated by these $n'$ columns has weight at least
	$k - n' + 1$. 
\end{itemize}
\end{lemma}

We now describe our general construction, using the so-called \emph{partial superregular}
matrices. 
\\

\nin\textbf{Main Construction.}
\begin{itemize}
	\item \textbf{Step 1.} Choose any superregular $k \times n$ matrix $\bG''$ and write it
		in the following block form. 
\begin{equation} 
\label{eq:V}
\bG'' = 
\left( \begin{array}{c|c|c}  	
    \smash{\overset{\mu}{\bA''}} & \smash{\overset{k - \mu}{\bB''}} & \smash{\overset{n-k}{\bC''}} \\
		\hline
    \bD & \bE & \bF
  \end{array} \right)\mbox{}^{\overset{k-\mu}{}}_{\underset{\mu}{}}
\end{equation} 
	\item \textbf{Step 2.} Perform elementary row operations to turn the matrix $\bA''$
	at the top-left corner into an all-zero matrix. We can do so by adding certain linear 
	combination of the last $\mu$ rows of $\bG''$ to each of its first $k - \mu$ rows. 
	Note that the $\mu \times \mu$ matrix $\bD$ is invertible, hence its rows can generate
	any vector of length $\mu$. The resulting matrix, referred to as $\bG'$, can be presented in block form as
	below. Since there is no row operation performed on the last $\mu$ rows, 
	the three block submatrices $\bD$, $\bE$, $\bF$ are the same in $\bG''$ and $\bG'$. \\
\begin{equation} 
\label{eq:M}
\bG' = 
\left( \begin{array}{c|c|c}  	
    \smash{\overset{\mu}{\bO{\color{white}'}}} & \smash{\overset{k - \mu}{\bB'}} & \smash{\overset{n-k}{\bC'}} \\
		\hline
    \bD & \bE & \bF
  \end{array} \right)\mbox{}^{\overset{k-\mu}{}}_{\underset{\mu}{}}
\end{equation} 
\item \textbf{Step 3.} Perform elementary row operations on the first $k-\mu$ rows
of $\bG'$ to turn the square submatrix $\bB'$ into a new square matrix $\bB$ of the same
size determined
as follows. It is a block diagonal matrix where each block submatrix is of size 
$p \times p$. Moreover, except from the zero entries, all other entries, which
belong to those block $p\times p$ submatrices, are the same as the corresponding
entries in $\bB''$. Equivalently, $\bB$ can be obtained from $\bB''$ by turning 
those entries that do not belong to any block diagonal submatrix into zero.
We can write $\bB$ in the block form as below, where $\bB''_i$ $(1 \leq i \leq
(k-\mu)/p)$ is the $i$th diagonal block $p \times p$ submatrix of $\bB''$. 
\begin{equation} 
\label{eq:B}
\bB =
\begin{blockarray}{ccccc}
p & p & \cdots & p & \\
\begin{block}{(cccc)c}
  \bB''_1 &  &  &  & p \\
   & \bB''_2 &  &  & p \\
   &  & \ddots &  & \vdots \\
   & &  & \bB''_{\frac{k-\mu}{p}} & p \\
\end{block}
\end{blockarray}
 \end{equation} 
Such transformation can always be done because both $\bB''$ and $\bB$ are invertible. 
Thus, the coding matrix $\bG$, as the output of Step 3, is determined by
\[
\bG = \bT \bG',
\]  
where the transform matrix $\bT$ is \\
\[
\bT = 
\left( \begin{array}{c|c}  	
    \smash{\overset{k-\mu}{\bB\bB'^{-1}}} & \smash{\overset{\mu}{\bO{\color{white}'}}} \\
		\hline
    \bO & \bI_\mu
  \end{array} \right)\mbox{}^{\overset{k-\mu}{}}_{\underset{\mu}{}}
\] 
We can write $\bG$ in block format as follows. 
\begin{figure}[H]
\centering
\includegraphics[scale=0.7]{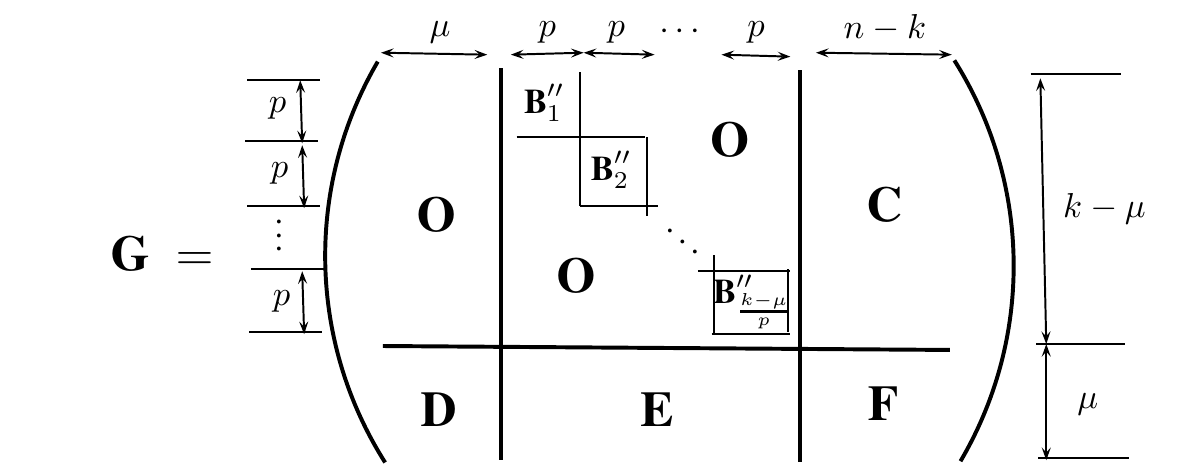}
\end{figure}
where $\bC = \bB\bB'^{-1}\bC'$. 
\end{itemize}

Note that the coding matrix $\bG$ produced by the
Main Construction has the same entries as the original superregular matrix $\bG''$, except 
for those zero entries and entries in block
submatrix $\bC$. Therefore, we often refer to
$\bG$ constructed in this way as \emph{partial}
superregular matrix. We illustrate the steps to construct a partial superregular matrix
in Example~\ref{ex:ps}. 

\begin{example}
\label{ex:ps}
Let $k = 5$, $n = 6$, $q = 11$, $\mu=1$, and $p = 2$. 

\nin\textbf{Step 1.} Let us choose $\bG''$ to be a Cauchy
matrix
\[
\bG''=
\left(
\begin{tabular}{c|cccc|c}
2 & 1 & 6 & 3 & 7 & 9\\
8 & 6 & 4 & 9 & 5 & 2\\
7 & 4 & 3 & 2 & 10 & 8\\
10 & 9 & 2 & 7 & 1 & 5\\
\hline
4 & 5 & 10 & 1 & 9 & 6
\end{tabular}
\right).
\]
\nin\textbf{Step 2.} As the bottom-left entry is nonzero, we can add certain multiple
of the last row to each of the first four
rows of $\bG''$ to obtain
\[
\bG'=
\left(
\begin{tabular}{c|cccc|c}
0 & 4& 1 & 8 & 8 & 6\\
0 & 7 & 6 & 7 & 9 & 1\\
0 & 9 & 2 & 3 & 8 & 3\\
0 & 2 & 10 & 10 & 6 & 1\\
\hline
4 & 5 & 10 & 1 & 9 & 6
\end{tabular}
\right).
\] 
\nin\textbf{Step 3.} Finally, we can perform
certain elementary row operations on the first four rows
to obtain the coding matrix
\[
\bG = 
\left(
\begin{tabular}{c|cccc|c}
0 & 1 & 6 & 0 & 0 & 7\\
0 & 6 & 4 & 0 & 0 & 0\\
0 & 0 & 0 & 2 & 10 & 7\\
0 & 0 & 0 & 7 & 1 & 1\\
\hline
4 & 5 & 10 & 1 & 9 & 6.
\end{tabular}
\right),
\]   
which is proved earlier in Example~\ref{ex:illustration}
to generate a (perfectly) $2$-decodable $1$-secure coding scheme. 
\end{example}

The following lemmas assert that a coding scheme based on a partial superregular matrix,
which is produced by the Main Construction, is indeed
$p$-decodable $\mu$-secure. 

\begin{lemma}
\label{lem:MDS}
The partial superregular matrix $\bG$ produced by
the Main Construction generates an $[n,k]$ MDS
code. Moreover, the last $\mu$ rows of $\bG$
also generates an $[n, \mu]$ MDS code. 
As a consequence, the coding scheme based on $\bG$
has the MDS property and moreover, satisfies (P1) (i.e., being strongly $\mu$-secure). 
\end{lemma}
\begin{proof}
As $\bG''$ generates an $[n,k]$ MDS code and $\bG$ is obtained
from $\bG''$ by applying only elementary row operation, $\bG$
generates the same $[n, k]$ MDS code. 
Since the last $\mu$ rows of $\bG$ are the same as those of $\bG''$, 
they also generates an $[n, \mu]$ MDS code. 
Hence (P1) is satisfied. 
\end{proof}

\begin{lemma}
\label{lem:decodability}
An erasure coding scheme based on a partial superreregular matrix $\bG$ (Main Construction)
always satisfies (P2): each subset of $p$ file symbols 
$\{s_{rp+1}, s_{rp+2}, \ldots, s_{(r+1)p}\}$ $(0 \leq r \leq (k-\mu)/p - 1)$
	can be reconstructed from the content of some $\mu + p$ storage nodes.
\end{lemma}
\begin{proof}
Recall that the file $\bs$ is encoded into $\bc = (\bs \mid \br)\bG$. 
We prove that the first $p$ file symbols $\{s_1, s_2, \ldots, s_p\}$ can be 
reconstructed from the first $\mu + p$ coordinates of $\bc$. 
In general, it can be proved in a similar manner that for each $r$ 
where $0 \leq r \leq (k-\mu)/p - 1$, the $p$ file symbols 
$\{s_{rp+1}, s_{rp+2}, \ldots, s_{(r+1)p}\}$
can be reconstructed from $\mu + p$ coordinates of $\bc$, namely
$c_1, \ldots, c_\mu$, and $c_{\mu + rp + 1}, \ldots, c_{\mu + (r+1)p}$. 

The coding matrix $\bG$ produced by the Main Construction can
be presented in the block form as follows. 
\[ 
\bG = 
\left(
\begin{tabular}{c|cccc|c}
\multirow{4}{*}{$\bO$} & $\bB''_1$ & & & & \multirow{4}{*}{$\bC$}\\
& & $\bB''_2$ & & & \\
& & & $\ddots$ & & \\
& & &  & $\bB''_{(k-1)/p}$ & \\
\hline
$\bD$ & $\bE_1$ & $\bE_2$ & $\cdots$ & $\bE_{(k-1)/p}$ & $\bF$
\end{tabular}
\right),
\]
where each $\bE_i$ is a $\mu\times p$ submatrix of $\bE$ such that 
$\bE = \big(\bE_1 \mid \bE_2 \mid \cdots \mid \bE_{(k-1)/p}\big)$. 
According to the structure of $\bG$, the first $\mu + p$ coordinates of $\bc$ can be written as
\[
(c_1 \cdots c_{\mu+p}) = (s_1 \cdots s_p \mid r_1 \cdots r_\mu)\bH_1,
\]
where 
\[
\bH_1 = 
\left(\begin{array}{c|c}  	
    \smash{\overset{\mu}{\bO{\color{white}'}}} & \smash{\overset{p}{\bB''_1}} \\
		\hline
    \bD & \bE_1
  \end{array}\right)\mbox{}^{\overset{p}{}}_{\underset{\mu}{}}.
\]
Since $\bG''$ is superregular, both $\bB''_1$ and $\bD$ are invertible. Let $\bv \in \fq^p$ be the $i$th column in the inverse
of $\bB''_1$ $(1 \leq i \leq p)$. Moreover, let $\bu \in \fq^\mu$ be the column vector satisfying $\bD \bu = -\bE_1 \bv$. 
Then
\[
\bH_1 
\left(\begin{array}{c}  	
    \bu \\
		\hline
    \bv 
  \end{array}\right)
	=
	\left(\begin{array}{c}  	
    \bB''_1\bv \\
		\hline
    \bD\bu + \bE_1\bv 
  \end{array}\right)
	=
	\left(\begin{array}{c}  	
    \be_i \\
		\hline
    \bzero 
  \end{array}\right),
\] 
where 
\[
\be_i = (\underbrace{0,\ldots,0}_{i-1},1,\underbrace{0,\ldots,0}_{p-i})^{\text{T}}. 
\]
Therefore, the user can reconstruct $s_i$ as follows.
\[
\begin{split}
s_i &= (s_1 \cdots s_p \mid r_1 \cdots r_\mu)
\left(\begin{array}{c}  	
    \be_i \\
		\hline
    \bzero 
  \end{array}\right)\\
	&= (s_1 \cdots s_p \mid r_1 \cdots r_\mu)\bH_1 
	\left(\begin{array}{c}  	
    \bu \\
		\hline
    \bv 
  \end{array}\right)\\
	&= (c_1 \cdots c_{\mu+p})
	\left(\begin{array}{c}  	
    \bu \\
		\hline
    \bv 
  \end{array}\right). 	
\end{split}
\] 
As $i$ can be chosen arbitrarily between $1$ and $p$, 
we conclude that an user which has access to the first
$\mu+p$ coordinates of $\bc$ can reconstruct all
$s_i$ for $1 \leq i \leq p$. The proof follows. 
\end{proof}
 
\begin{theorem}
\label{thm:imperfect}
The Main Construction produces a coding matrix that generates a $p$-decodable $\mu$-secure coding scheme
for all $0 \leq \mu < k$ and $1 \leq p \leq k - \mu$, 
$p | (k - \mu)$. 
\end{theorem}
\begin{proof}
According to Lemma~\ref{lem:MDS} and Lemma~\ref{lem:decodability}, the Main Construction
produces a coding matrix that generates a coding scheme satisfying
both (P1) and (P2). Hence, such scheme is
$p$-decodable $\mu$-secure, according to Definition~\ref{def:security_decodability}. 
\end{proof}

\section{On Perfectly $p$-Decodable $\mu$-Secure Coding Schemes}
\label{sec:perfect}

In this section, we first prove that a coding matrix produced by the Main Construction 
in Section~\ref{subsec:general} also satisfies (P3) when $p = 1$ $(\forall \mu)$ 
and when $\mu = 0, 1$ $(\forall p)$.
Finally, we establish the existence of perfectly $p$-decodable $\mu$-secure
coding schemes over sufficiently large fields for almost every $p$ and $\mu$
(namely, $k \geq 2(\mu+p)-1$).   

\subsection{The Case $p = 1$}
\label{subsec:p=1}

A (perfectly) $1$-decodable $\mu$-secure coding scheme is the best scheme among all strongly $\mu$-secure
schemes in terms of partial decoding. Such a scheme allows the user to reconstruct one file
symbol right after the user accesses one more node than the security threshold $\mu$. 
It is a sharp turn from knowing nothing about the file to being able to reconstruct one file symbol. 
In fact, according to Lemma~\ref{lem:decodability}, after accessing the first $\mu$ nodes, 
accessing any additional node would give the user one new file symbol. Hence, beyond the
threshold $\mu$, the coding scheme works in a similar manner as the conventional systematic coding scheme. 
Note that when $p=1$ and $\mu=0$, a (perfectly) $1$-decodable $0$-secure coding scheme is nothing
other than a normal systematic coding scheme. 

\begin{lemma}
When $p = 1$, the Main Construction produces a matrix $\bG$ that generates a (perfectly) $1$-decodable $\mu$-secure
coding scheme for any $\mu \geq 0$. 
\end{lemma}
\begin{proof}
When $p=1$, (P3) is satisfied trivially. Hence, the Main Construction yields a coding scheme
satisfying simultaneously (P1), (P2), and (P3). Such a coding scheme, according to Definition~\ref{def:security_decodability}, 
is perfectly $1$-decodable $\mu$-secure. 
\end{proof}

\subsection{The Case $\mu = 0$}
\label{subsec:muy=0}

A perfectly $p$-decodable $0$-secure coding scheme is of particular interest because of the following
properties.
\begin{itemize}
	\item \textbf{Weak security with no overhead on storage capacity:} the scheme provides
	weak security against an adversary which can access up to $p-1$ storage nodes. Such an adversary 
	cannot reconstruct any part of the stored file. Moreover, no storage overhead occurs
	compared to a normal erasure coding scheme: the file size is $k$ and the code is an $[n,k]$ MDS
	code. That is because there are no random symbols employed in the scheme.
	\item \textbf{$p$-partial decodability:} the user can reconstruct each subset of $p$ file 
	symbols by accessing certain $p$ storage nodes. 
\end{itemize}

\begin{lemma}
When $\mu=0$, the Main Construction produces a matrix $\bG$ that generates a perfectly $p$-decodable $0$-secure
coding scheme for any $p \geq 1$, $p | k$. 
\end{lemma}
\begin{proof}
Due to Lemma~\ref{lem:MDS} and Lemma~\ref{lem:decodability}, it suffices to prove that the coding matrix $\bG$ produced by the Main Construction 
generates a coding scheme satisfying (P3) - weak security. 
According to Lemma~\ref{lem:basic}, we aim to show that any set of $p-1\ (= \mu + p - 1)$ columns of $\bG$
does not generate an unit vector of weight one.
As $p | k$, we consider the following two cases: $k = p$ and $k \geq 2p$.

If $k=p$, as $\mu = 0$, the coding matrix $\bG$ is simply the same as the input superregular
matrix $\bG''$. 
By Lemma~\ref{lem:superregular}, any set of $p-1$ columns of $\bG$
generates a $[p,p-1,2]$ MDS code, hence never generates an 
unit vector $\be_i$, which has weight less than two. 
Therefore, according to Lemma~\ref{lem:basic}, the coding scheme
based on $\bG$ is weakly secure against an adversary which can access at most $p-1$ nodes. 
Hence (P3) is satisfied. 
%In fact, in this case, the scheme provides no partial decodability: 
%an user cannot reconstruct any part of the file by accessing less than $k$ nodes. 

We now assume that $k \geq 2p$. 
When $\mu = 0$, the matrix $\bG$ has the following form. 
\[
\bG = 
\left(
\begin{tabular}{cccc|c}
$\bB''_1$ & & & & \multirow{4}{*}{$\bC$}\\
& $\bB''_2$ & & & \\
& & $\ddots$ & & \\
& &  & $\bB''_{k/p}$ & 
\end{tabular}
\right).
\] 
We assume, by contradiction, that some set $L$ of $p-1$ columns of $\bG$
generates an unit vector $\be_i \in \fq^k$. Without loss of generality, we can assume that $i = 1$. For simplicity, we slightly abuse the notation and also use $L$ to
denote the set of indices of the columns in $L$. 
Then there exist some coefficients $\al_j \in \fq$ $(j \in L)$, so that
\begin{equation} 
\label{eq:1}
\sum_{j \in L} \al_j \bG[j] = \be_1 =
\left(
\begin{tabular}{c}
1\\
0\\
\vdots\\
0
\end{tabular}
\right)
,
\end{equation} 
where $\bG[j]$ denotes the $j$th column of $\bG$. 
Note that the $p \times p$ block matrix $\bB''_1$ is a
square submatrix of a superregular matrix $\bG''$,
and hence is invertible.
Therefore, there exists a linear combination of the first $p$ columns of $\bG$ that generate
the vector $-\be_1$. 
In other words, there exist some coefficients $\beta_j \in \fq$ $(j \in [p])$, such that
\begin{equation} 
\label{eq:2}
\sum_{j = 1}^p \beta_j \bG[j] = -\be_1 =
\left(
\begin{tabular}{c}
-1\\
0\\
\vdots\\
0
\end{tabular}
\right). 
\end{equation} 
By (\ref{eq:1}) and (\ref{eq:2}), we deduce that
\begin{equation} 
\label{eq:3}
\sum_{j = 1}^p \beta_j \bG[j]  + \sum_{j \in L} \al_j \bG[j] = \bzero.
\end{equation} 
Hence, there exists a linear combination of these at most $(2p-1)$ columns of $\bG$
that is equal to $\bzero$. Note that we assume $k \geq 2p$. Moreover, according to Lemma~\ref{lem:MDS}, 
$\bG$ is an $[n, k]$ MDS code. Hence, any set of at most $k$ columns of $\bG$ must be
linear independent. We obtain below a contradiction by arguing that the 
linear combination of at most $2p - 1 < k$ columns in (\ref{eq:3}) is \emph{nontrivial}, 
that is, their coefficients are not identically zero. Then the proof would follow.  

First, according to Lemma~\ref{lem:superregular},
as $\bB''_1$ is superregular, any $p-1$ columns of $\bB''_1$ form
a $[p, p - 1, 2]$ MDS code. In other words, any $p-1$ columns of $\bB''_1$ does not generate
a nonzero vector of weight less than two. Combining this fact with (\ref{eq:2}), we deduce
that $\beta_j \neq 0$ for all $j \in [p]$. Therefore, in the linear combination
(\ref{eq:3}), while the first sum consists of $p$ columns of $\bG$ with all nonzero coefficients, 
the second sum is a linear combination of $p-1$ columns of $\bG$. Hence, there must be at least
one term in the first sum that cannot be canceled out. Thus, (\ref{eq:3}) is a \emph{nontrivial} linear
combination of less than $k$ columns of $\bG$. This conclusion
contradicts the fact that any $k$ or less columns of $\bG$ must be linearly independent. 
\end{proof}

\subsection{The Case $\mu = 1$}
\label{subsec:mu=1}

\begin{lemma}
When $\mu=1$, the Main Construction produces a matrix $\bG$ that generates a perfectly $p$-decodable $1$-secure
coding scheme for any $p \geq 1$, $p | (k-1)$. 
\end{lemma}
\begin{proof}
Due to Lemma~\ref{lem:MDS} and Lemma~\ref{lem:decodability}, it suffices to prove that the coding matrix $\bG$ produced by the Main Construction 
generates a coding scheme satisfying (P3) - weak security. 
According to Lemma~\ref{lem:basic}, we aim to show that any set of $p\ (= \mu + p - 1)$ columns of $\bG$
does not generate an unit vector of weight one.
As $p | (k-1)$, there are two cases: $k - 1 = p$
and $k - 1 \geq 2p$. 

In the case $k - 1 = p$, there is
no partial decodability and we simply use the superregular
matrix $\bG''$ in Step 1 of the Main Construction, instead
of $\bG$. By Lemma~\ref{lem:superregular}, any set of $p = k-1$
columns of $\bG''$ generates a $[p,p-1,2]$ MDS code, hence
cannot generate an unit vector $\be_i$ of weight one. 
Therefore, the coding scheme based on $\bG''$ is weakly secure against an adversary which
can access at most $p-1$ storage nodes. Hence (P3) is 
satisfied. 

Now we assume that $k - 1\geq 2p$. 
When $\mu = 1$, the matrix $\bG$ has the following form. 
\[
\bG = 
\left(
\begin{tabular}{c|cccc|c}
\multirow{4}{*}{$\bzero$} & $\bB''_1$ & & & & \multirow{4}{*}{$\bC$}\\
& & $\bB''_2$ & & & \\
& & & $\ddots$ & & \\
& & &  & $\bB''_{(k-1)/p}$ & \\
\hline
$d_{1,1}$ & $\bE_1$ & $\bE_2$ & $\cdots$ & $\bE_{(k-1)/p}$ & $\bF$
\end{tabular}
\right),
\]
where $\bzero$ is a $(k-1) \times 1$ all-zero column vector,  $\bD = (d_{1,1})$ is a $1 \times 1$ matrix,
and $\bE_i$ is a $1 \times p$ row vector for each $1 \leq i
\leq (k-1)/p$. Note that $\bE = \big(\bE_1 \mid \bE_2 \mid 
\cdots \mid \bE_{(k-1)/p}\big)$. 

We now assume, by contradiction, that some set $L$ of 
$p$ columns of $\bG$ generates an unit vector
$\be_i \in \fq^k$, for some $i \in [k-1]$. 
Without loss of generality, we can assume that $i = 1$. 
For simplicity, we slightly abuse the notation and also use $L$ to
denote the set of indices of the columns in $L$. 
Then there exist some coefficients $\al_j \in \fq$ 
$(j \in L)$, so that
\begin{equation} 
\label{eq:4}
\sum_{j \in L} \al_j \bG[j] = 
\left(
\begin{tabular}{c}
$1$\\
$0$\\
$\vdots$\\
$0$\\
\hline
$0$\\
$\vdots$\\
$0$
\end{tabular}
\right)
,
\end{equation} 
where $\bG[j]$ denotes the $j$th column of $\bG$. 

Let $J = \{2,3,\ldots,p+1\}$.
We first argue that $L \neq J$. 
Indeed, the matrix
\begin{equation}
\label{eq:5}
\left( 
\begin{tabular}{c}
$\bB''_1$ \\ \hline $\bE_1$
\end{tabular}
\right)
\end{equation}
is a submatrix of the superregular matrix $\bG''$ chosen
in Step 1 of the Main Construction. Therefore, according
to Lemma~\ref{lem:superregular}, its columns generate
a $[p+1, p, 2]$ MDS code. Hence, every nontrivial linear
combination of these $p$ columns has weight at least two. 
As a consequence, any nontrivial linear combination
of $p$ columns $\bG[j]$'s $(j \in J)$, which correspond to the columns of the matrix in (\ref{eq:5}), has weight at least two. 
However, due to (\ref{eq:4}), 
the columns in $L$ can generate the unit vector $\be_1$, 
which has weight one. Thus, $L$ and $J$ must be different.

As $\bG''$ is superregular, $\bB''_1$ is invertible. 
Therefore, there exist some coefficients $\beta_j \in \fq$ $(j \in J)$, such that
\begin{equation} 
\label{eq:6}
\sum_{j \in J} \beta_j \bG[j] = 
\left(
\begin{tabular}{c}
$-1$ \\
$0$\\
$\vdots$\\
$0$\\
\hline
$\bE_1 \boldsymbol{\beta}^{\text{T}}$
\end{tabular}
\right), 
\end{equation} 
where $\boldsymbol{\beta} = (\beta_2, \ldots, \beta_{p+1})$.
According to Lemma~\ref{lem:superregular}, any set of less
than $p$ columns of $\bB''_1$ would never generate a nontrivial vector of weight less than two. 
Therefore, from (\ref{eq:6}), we deduce that $\beta_j \neq
0$ for every $j \in J$. 
Moreover, since $\bG''$ is superregular, $d_{1,1}$ must be nonzero.
Hence, there exists some $\gamma \in \fq$ such that 
$\gamma d_{1,1} = -\bE_1 \boldsymbol{\beta}^{\text{T}}$. Therefore, 
\begin{equation}
\label{eq:7}
\gamma \bG[1] =  
\left(
\begin{tabular}{c}
$0$ \\
$0$\\
$\vdots$\\
$0$\\
\hline
$-\bE_1 \boldsymbol{\beta}^{\text{T}}$
\end{tabular}
\right),
\end{equation}
From (\ref{eq:4}), (\ref{eq:6}), and (\ref{eq:7}) we have
\begin{equation}
\label{eq:8}
\sum_{j \in L} \al_j \bG[j] + \sum_{j \in J} \beta_j \bG[j]
+ \gamma \bG[1] = \bzero. 
\end{equation}
Note that the left-hand side of (\ref{eq:8}) is a linear
combination of at most $2p + 1\ (\leq k)$ columns of $\bG$.
We argue earlier that $\beta_j \neq 0$ for every
$j \in J$ and that $L \neq J$. 
Let $j_0 \in J \setminus L$. Then, the coefficient of $
\bG[j_0]$ in the linear combination of (\ref{eq:8}) is 
$\beta_{j_0} \neq 0$. Hence this linear combination is
nontrivial.   
Therefore, we obtain a contradiction, since $\bG$ generates
an $[n,k]$ MDS code and at the same time, has some set of
at most $k$ columns that are linearly dependent. 
\end{proof}
\vskip 10pt 

We summarize the results established so far on perfectly
$p$-decodable $\mu$-secure codes in the following theorem. 

\vskip 10pt  
\begin{theorem}
The Main Construction produces a coding matrix that 
generates a perfectly $p$-decodable $\mu$-secure coding scheme when $p = 1$ for every $\mu \geq 0$, $\mu < k$, and when $\mu = 0, 1$ for every $p \geq 1$, $p | (k-\mu)$. 
\end{theorem}

\subsection{Existence of Perfect Coding Schemes Over Large Fields}
\label{subsec:large_q}

For general $p$ and $\mu$, the Main Construction may not produce a 
perfectly $p$-decodable $\mu$-secure coding matrix. 
In order to achieve a perfect code, there is another property 
that the superregular matrix $\bG''$ (used in Step 1 of the Main Construction) 
must satisfy. We first define this property in Definition~\ref{def:perfect}
and then proceed to prove that there always exists a superregular matrix
satisfying this property as long as the field size is sufficiently large. 

\begin{definition}
\label{def:perfect}
A $k \times n$ matrix $\bG''$ $(k \leq n)$ is said to be \emph{$(p,\mu)$-superregular}
if it is superregular and also satisfies the following additional property. 
Let $\bG''$ be written in the block form 
\begin{equation} 
\label{eq:G''}
\bG'' = 
\left( \begin{array}{c|c|c}  	
    \smash{\overset{\mu}{\bA''}} & \smash{\overset{k - \mu}{\bB''}} & \smash{\overset{n-k}{\bC''}} \\
		\hline
    \bD & \bE & \bF
  \end{array} \right)\mbox{}^{\overset{k-\mu}{}}_{\underset{\mu}{}}.
\end{equation}
Let $\bB''_i$ be the $i$th $p\times p$ submatrix lying on the main diagonal of $\bB''$. Note that the entries outside these $p\times p$
block matrices are nonzeros. 
\[
\bB'' = 
\left(
\begin{tabular}{cccc}
$\bB''_1$ & & & \\
& $\bB''_2$ & & \\
& & $\ddots$ & \\
& & & $\bB''_{\frac{k-\mu}{p}}$
\end{tabular}
\right).
\] 
Let $\bE_i$ be the $i$th $\mu\times p$ submatrix of $\bE$ such that $\bE = (\bE_1 \mid \bE_2 \mid 
\cdots \mid \bE_{\frac{k-\mu}{p}})$. 
For each $i \in [(k-\mu)/p]$ we consider the block matrix $\bH_i$ given below.
\begin{equation}
\label{eq:H}
\bH_i = 
\left(\begin{array}{c|c}  	
    \smash{\overset{\mu}{\bO{\color{white}'}}} & \smash{\overset{p}{\bB''_i}} \\
		\hline
    \bD & \bE_i
  \end{array}\right)\mbox{}^{\overset{p}{}}_{\underset{\mu}{}}.
\end{equation}
Then it is required that for every $i \in [(k-\mu)/p]$, 
deleting simultaneously an arbitrary row among the first $p$ rows of $\bH_i$
and an arbitrary column among the first $\mu$ columns of $\bH_i$ always
results in an invertible $(\mu + p - 1) \times (\mu + p - 1)$ matrix. 
\end{definition}

Let $\XiG = (\xi_{i,j})$ be a $k \times n$ matrix where
$\xi_{i,j}$'s are indeterminates over some $\fq$. 
The subscript $_{\bG''}$ simply means that at a certain time, we will replace the entries $\xi_{i,j}$'s of 
$\XiG$ by appropriate elements of $\fq$ to obtain a matrix $\bG''$ over $\fq$, 
which is to be used in the Main Construction. 
We represent $\XiG$ as a block matrix as follows. 
\begin{equation}
\label{eq:XiG}
\XiG = 
\left( 
  \begin{array}{c|c|c}  	
    \smash{\overset{\mu}{\XiA}} & \smash{\overset{k - \mu}{\XiB}} & \smash{\overset{n-k}{\XiC}} \\
		\hline
    \XiD & \XiE & \XiF
  \end{array} 
\right)\mbox{}^{\overset{k-\mu}{}}_{\underset{\mu}{}}.
\end{equation}
Let $f_{\text{super}}(\ldots, \xi_{i,j}, \ldots)
\in \fq[\ldots, \xi_{i,j},\ldots]$
be the product of the determinants of all square submatrices
of $\XiG$. 
%Then obviously this polynomial is not identically zero: $f_{\text{super}} \not\equiv \bzero$.

Let $\XiBi$ be the $i$th $p \times p$ submatrix lying on 
the main diagonal of $\XiB$. 
\begin{equation}
\label{eq:XiB}
\XiB = 
\left(
\begin{tabular}{cccc}
$\XiBo$ & & & \\
& $\XiBt$ & & \\
& & $\ddots$ & \\
& & & $\XiBr$
\end{tabular}
\right).
\end{equation}
Let $\XiE$ be the $i$th $\mu \times p$ submatrix of $\XiE$ such that $\XiE = (\XiEo \mid \XiEt \mid 
\cdots \mid \XiEr)$. 
For each $i \in [(k-\mu)/p]$ we consider the block matrix $\XiHi$ given below.
\begin{equation}
\label{eq:XiH}
\XiHi = 
\left(\begin{array}{c|c}  	
    \smash{\overset{\mu}{\bO{\color{white}'}}} & \smash{\overset{p}{\XiBi}} \\
		\hline
    \XiD & \XiEi
  \end{array}\right)\mbox{}^{\overset{p}{}}_{\underset{\mu}{}}.
\end{equation}
Let $f_i(\ldots,\xi_{i,j},\ldots) \in \fq[\ldots, \xi_{i,j}, \ldots]$ be the product of determinants of all
square $(\mu+p-1)\times (\mu+p-1)$ submatrices of
$\XiHi$ obtained by deleting simulatenously one arbitrary
row among the first $p$ rows and one arbitrary column among
the first $\mu$ columns.
Let $\fmp = \prod_{i = 1}^{\frac{k-\mu}{p}}
f_i \in \fq[\ldots, \xi_{i,j}, \ldots]$. 

\begin{lemma}
\label{lem:auxiliary1}
The polynomials $\fs$ and $\fmp$ defined as above both are
not identically zero. 
\end{lemma}
\begin{proof}
In this proof we use the standard definition of
determinant to calculate a determinant of a matrix. 
More specifically, if $\bP = (p_{i,j})$ is any square matrix of order $r$ then the determinant 
\begin{equation}
\label{eq:det}
\det(\bP) = \sum_{\sigma \in \mathscr{S}_r} \sgn(\sigma) 
\prod_{i = 1}^r p_{i,\sigma(i)},
\end{equation}
where $\mathscr{S}_r$ denotes the symmetric (permutation) group on $r$ elements, and
$\sgn(\sigma)$ denotes the sign of the permutation $\sigma$. 

We first prove that $\fs$ is not identically zero. 
It suffices to show that the determinant of each square submatrix of $\XiG$ is not identically zero. 
According to (\ref{eq:det}), such a determinant is 
a sum of $(\mu+p-1)!$ monomials (with appropriate signs $\pm 1$) of $\xi_{i,j}$'s.
Those monomials are obviously pairwise distinct. Therefore, their sum is not identically zero. 

We now show that $\fmp$ is not identically zero. 
It suffices to show that for each $i \in [(k-\mu)/p]$, 
deleting simulateneously a row among the first $p$ rows
and a column among the first $\mu$ columns of $\XiHi$
(given in (\ref{eq:XiH})) results in a matrix with 
not identically zero determinant. 
For convenience, let us refer to such a matrix as 
$\bP$. Then $\bP$ is a square matrix of order $(\mu+p-1)$.
All entries of $\bP$ are indeterminates $\xi_{i,j}$, 
except for those lying in the $(p-1) \times (\mu-1)$
top-left submatrix. 
\begin{equation}
\label{eq:P}
\bP = 
\left( \begin{array}{c|c}  	
    \smash{\overset{\mu-1}{\bO}} & \smash{\overset{p}{(\xi_{i,j})}}\\
		\hline
    (\xi_{i,j}) & (\xi_{i,j})
  \end{array} \right)\mbox{}^{\overset{p-1}{}}_{\underset{\mu}{}}.
\end{equation}
Therefore, all entries on the antidiagonal (the diagonal goes from the lower-left corner
to the upper-right corner of the matrix) are nonzero. 
These $\mu+p-1$ antidiagonal entries form an unique nonzero monomial of $\xi_{i,j}$'s in the 
formula for determinant of $\bP$ in (\ref{eq:det}), which cannot be
canceled out by any other monomial. Hence, $\det(\bP)$ is not identically zero as a polynomial. 
\end{proof}

\begin{lemma}
\label{lem:auxiliary2}
For every $0 \leq \mu < k$ and $1 \leq p \leq k - \mu$, 
$p | (k - \mu)$, 
there always exists a $(p, \mu)$-superregular $k \times n$ matrix $\bG''$
over any sufficiently large field. 
\end{lemma}
\begin{proof}
According to Lemma~\ref{lem:auxiliary1}, the polynomial $\fs \fmp \in \fq[\ldots,\xi_{i,j},\ldots]$
is not identically zero. 
By \cite[Lemma 4]{Ho2006}, for sufficiently large $q$, there exist $g''_{i,j} \in \fq$
such that $\fs(\ldots, g''_{i,j}, \ldots) \neq 0$ and $\fmp (\ldots, g''_{i,j}, \ldots) \neq 0$. 
The former condition guarantees that the $k \times n$ matrix $\bG''=(g''_{i,j})$ is a superregular
matrix, while the latter further implies that $\bG''$ is $(p,\mu)$-superregular, according to 
Definition~\ref{def:perfect}.
Thus, a $(p, \mu)$-superregular $k \times n$ matrix $\bG''$ always exists 
over any sufficiently large field.   
\end{proof}

\begin{lemma}
\label{lem:auxiliary3}
Suppose that the matrix $\bG''$ used in Step 1 of the Main Construction
is $(p,\mu)$-superregular and that $\bG$ is the resulting coding matrix.
\begin{equation} 
\label{eq:G1}
\bG = 
\left(
\begin{tabular}{c|cccc|c}
\multirow{4}{*}{$\bO$} & $\bB''_1$ & & & & \multirow{4}{*}{$\bC$}\\
& & $\bB''_2$ & & & \\
& & & $\ddots$ & & \\
& & &  & $\bB''_{(k-1)/p}$ & \\
\hline
$\bD$ & $\bE_1$ & $\bE_2$ & $\cdots$ & $\bE_{(k-1)/p}$ & $\bF$
\end{tabular}
\right),
\end{equation}
where each $\bE_i$ is a $\mu\times p$ submatrix of $\bE$ such that 
$\bE = \big(\bE_1 \mid \bE_2 \mid \cdots \mid \bE_{(k-1)/p}\big)$. 
For each $i \in [(k-\mu)/p]$ let  
\begin{equation}
\label{eq:Hi}
\bH_i = 
\left(\begin{array}{c|c}  	
    \smash{\overset{\mu}{\bO{\color{white}'}}} & \smash{\overset{p}{\bB''_i}} \\
		\hline
    \bD & \bE_i
  \end{array}\right)\mbox{}^{\overset{p}{}}_{\underset{\mu}{}}.
\end{equation}
Then any set of $\mu+p-1$ columns of $\bH_i$ that consists of
the last $p$ columns and some $\mu-1$ columns among the first $\mu$
columns generates a $[\mu+p, \mu+p-1,2]$ MDS code. 
Thus, such a set of $\mu+p-1$
columns of $\bH_i$, as well as the corresponding set of $\mu+p-1$
columns of $\bG$, never generates a nontrivial vector of 
weight less than two.
\end{lemma} 
\begin{proof} 
Note that since $\bG''$ is $(p,\mu)$-superregular,
deleting simultaneously an arbitrary row among the first $p$ rows
and an arbitrary column among the first $\mu$ columns of $\bH_i$
always results in an invertible matrix of order $\mu+p-1$.

Let $\bK$ be the submatrix of $\bH_i$ obtained by deleting
some column among the first $\mu$ columns of $\bH_i$. 
In order to show that the columns of $\bK$ generate an MDS code, 
it suffices to prove that every square submatrix $\bQ$ of order $\mu+p-1$
of $\bK$ is invertible. Indeed, if $\bQ$ is
obtained by deleting one row among the first $p$ rows of $\bK$, 
then it is obviously invertible, due to the property of $\bH_i$
specified earlier. Suppose that $\bQ$ is obtained by 
deleting one row among the last $\mu$ rows of $\bK$. 
Then $\det(\bQ)$ is equal to the product of $\det(\bB''_i)$ and
the determinant of a square submatrix of order $\mu-1$ of $\bD$. 
As both $\bB''_i$ and $\bD$ are superregular, we conclude that the determinant of $\bQ$ must be nonzero. Therefore, $\bQ$ is invertible.  
\end{proof} 

\begin{theorem} 
For every $0 \leq \mu < k$ and $1 \leq p \leq k - \mu$, $p | (k - \mu)$,
such that $k \geq 2(\mu+p)-1$,  
there exists a perfectly $p$-decodable $\mu$-secure coding scheme over
any sufficiently large field $\fq$. 
\end{theorem}
\begin{proof}
In Step 1 of the Main Construction let us choose $\bG''$ to be a $(p,\mu)$-superregular
matrix over a sufficiently large field $\fq$. 
Let $\bG$ be the coding matrix produced by the Main Construction.
We present $\bG$ in the block form as shown in the Main Construction.
\begin{equation} 
\label{eq:G}
\bG = 
\left(
\begin{tabular}{c|cccc|c}
\multirow{4}{*}{$\bO$} & $\bB''_1$ & & & & \multirow{4}{*}{$\bC$}\\
& & $\bB''_2$ & & & \\
& & & $\ddots$ & & \\
& & &  & $\bB''_{(k-1)/p}$ & \\
\hline
$\bD$ & $\bE_1$ & $\bE_2$ & $\cdots$ & $\bE_{(k-1)/p}$ & $\bF$
\end{tabular}
\right). 
\end{equation} 
Note that each $\bE_i$ is a $\mu\times p$ submatrix of $\bE$ such that 
$\bE = \big(\bE_1 \mid \bE_2 \mid \cdots \mid \bE_{(k-1)/p}\big)$. 
  
According to Lemma~\ref{lem:MDS} and Lemma~\ref{lem:decodability},
$\bG$ satisfies (P1) and (P2). 
It remains to show that $\bG$ satisfies (P3) - weak security. 

According to Lemma~\ref{lem:basic}, we aim to show that no subsets
of $(\mu+p-1)$ columns of $\bG$ generate the unit vector $\be_i$
for every $i \in [k-\mu]$. 
Without loss of generality, we can assume, by contradiction, that a subset
$L$ of some $(\mu+p-1)$ columns of $\bG$ generates the unit vector $\be_1$. 
For simplicity, we slightly abuse the notation and also use $L$ to
denote the set of indices of the columns in $L$. 
Then there exist some coefficients $\al_j \in \fq$ $(j \in L)$ such that
\begin{equation}
\label{eq:l1}
\sum_{j \in L} \al_j \bG[j] = 
\left(
\begin{tabular}{c}
$1$ \\
$0$\\
$\vdots$\\
$0$\\
\hline
$0$\\
$\vdots$\\
$0$
\end{tabular}
\right),
\end{equation}  
where $\bG[j]$ is the $j$th column of $\bG$. 
For convenience, we assign $\al_j = 0$ for all $j \in [n] \setminus L$. 
Let $J = \{\mu+1,\mu+2,\ldots,\mu+p\}$. 

As $\bB''_1$ is invertible, there exist some coefficients $\bt_j \in \fq$ such that
\begin{equation}
\label{eq:l2}
\sum_{j \in J} \beta_j \bG[j] = 
\left(
\begin{tabular}{c}
$-1$ \\
$0$\\
$\vdots$\\
$0$\\
\hline
$\bE_1 \boldsymbol{\beta}^{\text{T}}$
\end{tabular}
\right), 
\end{equation} 
where $\boldsymbol{\beta} = (\beta_{\mu+1}, \ldots, \beta_{\mu+p})$.
According to Lemma~\ref{lem:superregular}, any set of less
than $p$ columns of $\bB''_1$ would never generate a nontrivial vector of weight less than two. 
Therefore, from (\ref{eq:l2}), we deduce that $\beta_j \neq
0$ for every $j \in J$. 
Moreover, as $\bD$ is invertible, 
there exist some coefficients $\gamma_j \in \fq$ such that 
\[
\sum_{j \in [\mu]}\gamma_j \bD[j] = -\bE_1 \boldsymbol{\beta}^{\text{T}},
\]
where $\bD[j]$ is the $j$th column of $\bD$. 
Therefore, 
\begin{equation}
\label{eq:l3}
\sum_{j \in [\mu]}\gamma_j \bG[j] =  
\left(
\begin{tabular}{c}
$0$ \\
$0$\\
$\vdots$\\
$0$\\
\hline
$-\bE_1 \boldsymbol{\beta}^{\text{T}}$
\end{tabular}
\right).
\end{equation}

From (\ref{eq:l1}), (\ref{eq:l2}), and (\ref{eq:l3}) we have
\begin{equation}
\label{eq:l4}
\sum_{j \in L} \al_j \bG[j] + \sum_{j \in J} \beta_j \bG[j]
+ \sum_{j \in [\mu]} \gamma_j \bG[j] = \bzero. 
\end{equation}
Note that the left-hand side of (\ref{eq:l4}) is a linear
combination of at most $2(\mu+p) - 1\ (\leq k)$ columns of $\bG$.
We consider the following three cases and aim to obtain contradictions in all cases. \\

\nin\textbf{Case 1.}
$\exists j_0 \in J: \al_{j_0} = 0$.

We argue earlier that $\beta_j \neq 0$ for every
$j \in J$. Moreover, in this case, there exists $j_0 \in J$ such that $\al_{j_0} = 0$. 
Hence, the linear combination in (\ref{eq:l4}) must be
nontrivial, since at least one vector, namely $\bG[j_0]$, has a nonzero coefficient $\beta_{j_0}$.   
Therefore, we obtain a contradiction, since $\bG$ generates
an $[n,k]$ MDS code and at the same time, has some set of
at most $k$ columns that are linearly dependent.\\

\nin\textbf{Case 2.}
$\forall j \in J: \al_j \neq 0$ and $L \setminus J \subseteq [\mu]$. 

Note that in this case, $L$ consists of $p$ columns indexed by $J$ and 
some $\mu-1$ of the first $\mu$ columns of $\bG$. 
According to Lemma~\ref{lem:auxiliary3}, $L$ does not generate
any nontrivial vector of weight less than two. 
This conclusion contradicts (\ref{eq:l1}).
\\
 
\nin\textbf{Case 3.}
$\forall j \in J: \al_j \neq 0$ and $L \setminus J \not\subseteq [\mu]$. 

In this case, $L$ consists of $p$ columns indexed by $J$, 
at most $\mu-1$ of the first $\mu$ columns of $\bG$, and
probably some columns indexed by elements in the set $[n] \setminus ([\mu] \cup J)$. 
If $\al_j = 0$ for all $j \in L \setminus ([\mu] \cup J)$,
then similar to Case 2, we obtain a contradiction due to 
Lemma~\ref{lem:auxiliary3} and (\ref{eq:l1}).

If there exists $j \in L \setminus ([\mu] \cup J)$ such that $\al_j 
\neq 0$, then (\ref{eq:l4}) presents a nontrivial linear combination of at most $k$ columns of $\bG$, which is a zero vector.
This assertion contradicts the fact that as $\bG$ generates
an $[n,k]$ MDS code, any set of $k$ columns of $\bG$ must 
be linearly independent.   
\end{proof}

\section{Conclusion}
\label{sec:conclusion}

We propose in this work a method to construct erasure coding
schemes which are not only (strongly and weakly) secure but also
partially decodable. The partial decodablity feature 
is extremely important in applications such as media
streaming, where it is usually not necessarily for the user to download the whole file
before he or she can start the playback. 
 
The type of erasure coding scheme developed in our work
offer the flexibility between security and partial decodability. 
The system designer can adjust the security parameter and the
partial decodability parameter accordingly to obtain any possible 
mixture of security and exposure of the stored data. 
We emphasize that we can construct an erasure code which is 
both secure and partially decodable without adding any extra storage
overhead compared to a merely secure erasure code studied in the literature.  

\bibliographystyle{ieeetran}
\bibliography{Quasi-Systematic}

\begin{thebibliography}{10}
\providecommand{\url}[1]{#1}
\csname url@samestyle\endcsname
\providecommand{\newblock}{\relax}
\providecommand{\bibinfo}[2]{#2}
\providecommand{\BIBentrySTDinterwordspacing}{\spaceskip=0pt\relax}
\providecommand{\BIBentryALTinterwordstretchfactor}{4}
\providecommand{\BIBentryALTinterwordspacing}{\spaceskip=\fontdimen2\font plus
\BIBentryALTinterwordstretchfactor\fontdimen3\font minus
  \fontdimen4\font\relax}
\providecommand{\BIBforeignlanguage}[2]{{%
\expandafter\ifx\csname l@#1\endcsname\relax
\typeout{** WARNING: IEEEtran.bst: No hyphenation pattern has been}%
\typeout{** loaded for the language `#1'. Using the pattern for}%
\typeout{** the default language instead.}%
\else
\language=\csname l@#1\endcsname
\fi
#2}}
\providecommand{\BIBdecl}{\relax}
\BIBdecl

\bibitem{MW_S}
F.~J. MacWilliams and N.~J.~A. Sloane, \emph{The Theory of Error-Correcting
  Codes}.\hskip 1em plus 0.5em minus 0.4em\relax Amsterdam: North-Holland,
  1977.

\bibitem{WeatherspoonKubiatowicz2002}
H.~Weatherspoon and J.~Kubiatowicz, ``Erasure coding vs. replication: A
  quantitative comparison,'' in \emph{Revised Papers from the First
  International Workshop on Peer-to-Peer Systems}, ser. IPTPS '01, 2002, pp.
  328--338.

\bibitem{Huang2012}
C.~Huang, H.~Simitci, Y.~Xu, A.~Ogus, B.~Calder, P.~Gopalan, J.~Li, and
  S.~Yekhanin, ``Erasure coding in windows azure storage,'' in \emph{Proc.
  USENIX Conf. Annual Technical Conference (ATC)}, 2012, pp. 2--2.

\bibitem{Thusoo2010}
A.~Thusoo, Z.~Shao, S.~Anthony, D.~Borthakur, N.~Jain, J.~S. Sarma, R.~Murthy,
  and H.~Liu, ``Data warehousing and analytics infrastructure at {F}acebook,''
  in \emph{Proc. ACM SIGMOD Int. Conf. on Management of Data (SIGMOD)}, 2010,
  pp. 1013--1020.

\bibitem{Ford2010}
D.~Ford, F.~Labelle, F.~I. Popovici, M.~Stokely, V.-A. Truong, L.~Barroso,
  C.~Grimes, and S.~Quinlan, ``Availability in golobally distributed storage
  systems,'' in \emph{Proc. USENIX Conf. Operating Syst. Design
  Implementation}, 2010, pp. 1013--1020.

\bibitem{OzarowWyner1984}
L.~H. Ozarow and A.~D. Wyner, ``The wire-tap channel {II},'' \emph{Bell Syst.
  Tech. J.}, vol.~63, pp. 2135--2157, 1984.

\bibitem{SubramanianMcLaughlin2009}
\BIBentryALTinterwordspacing
A.~Subramanian and S.~W. McLaughlin, ``{MDS} codes on the erasure-erasure
  wiretap channel,'' 2009. [Online]. Available:
  \url{http://arxiv.org/abs/0902.3286}
\BIBentrySTDinterwordspacing

\bibitem{Yamamoto1986}
H.~Yamamoto, ``Secret sharing system using (k, {L}, n) threshold scheme,''
  \emph{Electronics and Communications in Japan (Part I: Communications)},
  vol.~69, no.~9, pp. 46--54, 1986.

\bibitem{Yamamoto2006}
M.~Iwamoto and H.~Yamamoto, ``Strongly secure ramp secret sharing schemes for
  general access structures,'' \emph{Inf. Process. Lett.}, vol.~97, no.~2, pp.
  52--57, 2006.

\bibitem{OliveiraLimaVinhozaBarrosMedard2012}
P.~F. Oliveira, L.~Lima, T.~T.~V. Vinhoza, J.~Barros, and M.~M\'{e}dard,
  ``Coding for trusted storage in untrusted networks,'' \emph{IEEE Trans.
  Inform. Forensics Security}, vol.~7, no.~6, pp. 1890--1899, 2012.

\bibitem{DauSongYuen2014}
S.~H. Dau, W.~Song, and C.~Yuen, ``On block security of regenerating codes at
  the {MBR} point for distributed storage systems,'' in \emph{Proc. Int. IEEE
  Symp. Inform. Theory (ISIT)}, 2014, pp. 1967--1971.

\bibitem{KadheSprintson2014}
S.~Kadhe and A.~Sprintson, ``Weakly secure regenerating codes for distributed
  storage,'' in \emph{Int. Symp. Network Coding (NetCod)}, 2014, pp. 1--6.

\bibitem{CaiChan2011}
N.~Cai and T.~H. Chan, ``Theory of secure network coding,'' \emph{Proceedings
  of the IEEE}, vol.~99, no.~3, pp. 421--437, 2011.

\bibitem{BhattadNarayanan2005}
K.~Bhattad and K.~R. Narayanan, ``Weakly secure network coding,'' in
  \emph{Proc. 1st Workshop on Network Coding, Theory, and Application
  (NetCod)}, 2005.

\bibitem{RothLempel1989}
R.~M. Roth and A.~Lempel, ``On {MDS} codes via {C}auchy matrices,'' \emph{IEEE
  Trans. Inform. Theory}, vol.~35, no.~6, pp. 1314--1319, 1989.

\bibitem{LacanFimes2004}
J.~Lacan and J.~Fimes, ``Systematic {MDS} erasure codes based on {V}andermonde
  matrices,'' \emph{IEEE Commun. Lett.}, vol.~8, no.~9, pp. 570--572, 2004.

\bibitem{Ho2006}
T.~Ho, M.~M\'{e}dard, R.~Koetter, D.~R. Karger, M.~Effros, J.~Shi, and
  B.~Leong, ``A random linear network coding approach to multicast,''
  \emph{IEEE. Trans. Inform. Theory}, vol.~52, no.~10, 4413--4430.

\end{thebibliography}

\appendix
\subsection{Proof of Lemma~\ref{lem:basic}}
\label{app:1}

The \emph{only if} direction is obvious, according to the
discussion below Lemma~\ref{lem:basic}. It remains to prove
the \emph{if} direction. 

Let $\bst$ and $\brt$ be the guessed value (by the 
adversary) for the real file $\bs$ and for the random vector $\br$, respectively. By accessing $m$
storage nodes, the adversary obtains some linear combinations of the file symbols $s_i$'s and the random
symbols $r_j$'s, namely $\ba = (\bs\mid\br)\bM$. 
We aim to show that for every $i 
\in [k-\mu]$, all possible values for $s_i$ are equally
probable.  

Choose an arbitrary $i \in [k-\mu]$ and an arbitrary
element $\vt \in \fq$. It suffices to show that the 
system of linear equations
\begin{equation} 
\label{eq:app1}
\begin{cases}
(\bst \mid \brt)\bM = \ba\\
(\bst \mid \brt)\be_i^{\text{T}} = \vt
\end{cases}
\end{equation} 
always has the same number of solutions $(\bst \mid \brt)$ for every 
choice of $\vt \in \fq$. 
Here $\be_i$ denotes the unit vector with a one
at the $i$th coordinate and zeroes elsewhere. 
It is a basic fact from linear
algebra that the solution set for the system (\ref{eq:app1}) 
above, if nonempty, is an affine space, which is the sum of one solution of (\ref{eq:app1}) and the solution space
of the corresponding homogeneous system. 
Therefore, if the system (\ref{eq:app1}) always has at least one solution for every $\vt$, then it would have the same number of solutions for every $\vt$. 
Therefore, it remains to prove that this system always has a solution for every choice of $\vt \in \fq$. 

Note that we have
\begin{equation} 
\label{eq:app2}
\begin{cases}
(\bs \mid \br)\bM = \ba\\
(\bs \mid \br)\be_i^{\text{T}} = s_i
\end{cases}.
\end{equation} 
By subtracting the corresponding equations in the two
systems (\ref{eq:app1}) and (\ref{eq:app2}) and let
$\bx = (\bst - \bs \mid \brt - \br)$ be the new unknowns,
we obtain the following system
\begin{equation} 
\label{eq:app3}
\begin{cases}
\bx\bM = \bzero\\
\bx\be_i^{\text{T}} = \vt - s_i
\end{cases}
\end{equation}
It is clear that the system (\ref{eq:app1}) has a solution
if and only if the system (\ref{eq:app3}) has a solution. 
Hence, it suffices to show that the system (\ref{eq:app3})
always has a solution for every choice of $\vt \in \fq$. 

We claim that there exists some $\overline{\bx} \in
\fq^k$ satisfying $\overline{\bx}\bM = \bzero$ and 
$\bx\be_i^{\text{T}} \neq 0$. Then it is obvious that
\[
\bx^* = \dfrac{\vt-s_i}{\bx\be_i^{\text{T}}}\overline{\bx}
\]
would be a solution of (\ref{eq:app3}), and hence the proof
follows. Indeed, if $\bx\be_i^{\text{T}} = 0$ for every
$\bx$ satisfying $\bx \bM = \bzero$ then $\be_i^{\text{T}}$
must belong the the orthogonal complement of the solution
space of the system $\bx \bM = \bzero$, which is precisely
the column space of $\bM$. However, according to our
assumption, the column space of $\bM$ does not contain
$\be_i$. Thus, there must exists some $\overline{\bx} \in \fq^k$ satisfying $\overline{\bx}\bM = \bzero$ and 
$\bx\be_i^{\text{T}} \neq 0$, as claimed above. 
\end{document}